\documentclass{new_tlp}
\usepackage{mathptmx}

\usepackage{mathtools}
\DeclarePairedDelimiter\ceil{\lceil}{\rceil}

\usepackage{xspace}
\usepackage[ruled,linesnumbered,vlined]{algorithm2e}\SetProcNameSty{textsc}
\usepackage{subfigure}
\usepackage{epsfig}
\usepackage{epstopdf}
\usepackage{comment}
\usepackage{amssymb}
\usepackage{amsmath}
\usepackage[T1]{fontenc}
\usepackage{url}
\usepackage{graphicx,amssymb}
\usepackage[textsize=small]{todonotes}
\usepackage{booktabs}
\usepackage{multirow}
\usepackage{tikz}
\usepackage{pgfplots}\usetikzlibrary{plotmarks}

\pgfplotsset{
	filter discard warning=false 
	, legend cell align=left
	, minor grid style={loosely dotted, lightgray}
	, major grid style={loosely dashed, lightgray}
}

\def\A{\ensuremath{\mathcal{A}}\xspace}
\def\naf{\ensuremath{\raise.17ex\hbox{\ensuremath{\scriptstyle\mathtt{\sim}}}}\xspace}
\def\COUNT{\ensuremath{\textsc{count}}\xspace}
\def\SUM{\ensuremath{\textsc{sum}}\xspace}
\def\cla{\reflectbox{\ \ensuremath{\leadsto}\ }\xspace}
\def\W{\ensuremath{\mathcal{W}}\xspace}
\def\SM{\ensuremath{\mathit{SM}}\xspace}
\def\OSM{\ensuremath{\mathit{OSM}}\xspace}

  \title[Theory and Practice of Logic Programming]
        {Anytime answer set optimization via unsatisfiable core shrinking}

  \author[M. Alviano and C. Dodaro]
         {MARIO ALVIANO and CARMINE DODARO\\
         Department of Mathematics and Computer Science, University of Calabria, Italy\\
         \email{\{alviano,dodaro\}@mat.unical.it}}

\jdate{March 2003}
\pubyear{2003}
\pagerange{\pageref{firstpage}--\pageref{lastpage}}

\newtheorem{theorem}{Theorem}

\newtheorem{example}{Example}

\newcommand{\clasp}{\textsc{clasp}\xspace}
\newcommand{\wasp}{\textsc{wasp}\xspace}
\newcommand{\waspdisj}{\textsc{wasp+disj}\xspace}

\begin{document}
\label{firstpage}

\maketitle

\begin{abstract}
Unsatisfiable core analysis can boost the computation of optimum stable models for logic programs with weak constraints.
However, current solvers employing unsatisfiable core analysis either run to completion, or provide no suboptimal stable models but the one resulting from the preliminary disjoint cores analysis.
This drawback is circumvented here by introducing a progression based shrinking of the analyzed unsatisfiable cores.
In fact, suboptimal stable models are possibly found while shrinking unsatisfiable cores, hence resulting into an anytime algorithm.
Moreover, as confirmed empirically, unsatisfiable core analysis also benefits from the shrinking process in terms of solved instances.
\end{abstract}



\section{Introduction}

Answer set programming (ASP) is a declarative formalism for knowledge representation and reasoning based on stable model semantics \cite{DBLP:journals/ngc/GelfondL91,DBLP:journals/amai/Niemela99,DBLP:journals/tplp/MarekNT08,DBLP:conf/aaai/Lifschitz08,DBLP:conf/rweb/EiterIK09,DBLP:journals/cacm/BrewkaET11}, for which efficient implementations such as \textsc{clasp} \cite{DBLP:conf/lpnmr/GebserKK0S15}, \textsc{cmodels} \cite{DBLP:conf/lpnmr/LierlerM04,DBLP:journals/jar/GiunchigliaLM06}, \textsc{dlv} \cite{DBLP:conf/datalog/AlvianoFLPPT10}, and \textsc{wasp} \cite{DBLP:conf/lpnmr/AlvianoDFLR13,DBLP:conf/cilc/DodaroAFLRS11} are available.
ASP programs are associated with classical models satisfying a stability condition:
only necessary information is included in a model of the input program under the assumptions provided by the model itself for the \emph{unknown knowledge} in the program, where unknown knowledge is encoded by means of \emph{default negation}. 
Moreover, the language includes several constructs to ease the representation of real world knowledge, among them aggregates \cite{DBLP:journals/ai/SimonsNS02,DBLP:journals/ai/LiuPST10,DBLP:conf/aaaiss/BartholomewLM11,DBLP:journals/ai/FaberPL11,DBLP:journals/tocl/Ferraris11,DBLP:conf/lpnmr/AlvianoF13,DBLP:journals/tplp/GelfondZ14,DBLP:journals/tplp/AlvianoL15,DBLP:journals/tplp/AlvianoFG15,DBLP:conf/aaai/AlvianoFS16}.

Reasoning in presence of unknown knowledge is quite common for rational agents acting in the real world.
It is also common that real world agents cannot meet all their desiderata, and therefore ASP programs may come with \emph{weak constraints} for representing numerical preferences over jointly incompatible conditions \cite{DBLP:journals/tkde/BuccafurriLR00,DBLP:journals/ai/SimonsNS02,DBLP:conf/iclp/GebserKKS11,DBLP:journals/tplp/GebserKS11}.
Stable models are therefore associated with a cost given by the sum of the weights of the unsatisfied weak constraints, so that stable models of minimum cost are \emph{preferred}.
It is important here to stress the meaning of the word preferred:
any stable model describes a plausible scenario for the knowledge represented in the input program, even if it may be only an admissible solution of non optimum cost.
In fact, many rational agents would still accept suboptimal solutions, possibly with an estimate on the maximum distance to the optimum cost.
This flexibility is also justified by the intrinsic complexity of the problem:
the computation of an optimum stable model requires in general at least linearly many calls to a $\Sigma^P_2$ oracle \citeA{DBLP:journals/tkde/BuccafurriLR00}, and it is therefore practically unfeasible for the hardest instances.

According to the above observations, a good algorithm for answer set optimization should produce better and better stable models during the computation of an optimum stable model.
Algorithms having this property are called \emph{anytime} in the literature \cite{DBLP:journals/tplp/AlvianoDR14,DBLP:conf/ijcai/BliemKSW16}, and among them there is \emph{linear search sat-unsat}, the first algorithm introduced for answer set optimization \cite{DBLP:journals/tkde/BuccafurriLR00}.
In short, any stable model is searched so to initialize an overestimate of the optimum cost, and new stable models of improved cost are iteratively searched until an unsatisfiability arises, proving optimality of the latest found stable model.
The algorithm is simple and anytime, but in practice quite inefficient in computing stable models of reasonable cost.
Moreover, the algorithm does not provide any underestimate, or error estimation.
(Section~\ref{sec:linear} provides additional details within this respect.)

Answer set optimization is often boosted by algorithms based on \emph{unsatisfiable core} analysis, such as \textsc{oll} \cite{DBLP:conf/iclp/AndresKMS12} and \textsc{one} (introduced by \citeNP{Alviano01092015} for MaxSAT, and adapted to ASP in Section~\ref{sec:one}).
These algorithms start by searching a stable model satisfying all weak constraints, which would be therefore an optimum stable model.
On the other hand, if there is no stable model of this kind, a subset of the weak constraints that cannot be jointly satisfied is identified.
Such a set is called unsatisfiable core, and essentially evidences that any optimum stable model must sacrifice at least one of the desiderata expressed by the weak constraints.
Hence, unsatisfiable cores provide underestimates of the cost of optimum stable models.
Moreover, the program can be modified by replacing the weak constraints in the core with new weak constraints that essentially express a preference for stable models satisfying all but one of the original weak constraints, and anyhow the largest number of them, so that the process can be reiterated.

The main drawback of these algorithms is that they run to completion, and therefore provide no intermediate suboptimal stable models for the hardest instances that cannot be solved in reasonable time.
In ASP, current alternatives are based on \emph{disjoint cores analysis}, and on \emph{stratification}.
Disjoint cores analysis amounts to identifying unsatisfiable cores with empty intersection, and terminates with a (suboptimal) stable model.
After that, the usual core based algorithm is run to completion, providing no further intermediate stable models.
With stratification, instead, only weak constraints of greatest weight are first considered, so that the underestimate may be better improved by the identified unsatisfiable cores.
After all cores have been analyzed, a (suboptimal) stable model will be found, and weak constraints with the second greatest weight will be also considered.
The process is repeated, hence producing a (suboptimal) stable model at each iteration, possibly improving the current overestimate.
However, no intermediate stable model is produced before completing each iteration, which means that few overestimate improvements occur when few different weights occur in the input.
Actually, when all weak constraints have the same weight, the algorithm runs to completion providing no intermediate stable models.

The main questions addressed in this paper are therefore the following:
How to produce better and better stable models while running core based algorithms for answer set optimization?
Should any overhead be paid for obtaining such an anytime behavior?
The answer we propose for the first question is the following:
Unsatisfiable cores are often non-minimal, and their sizes can be significantly reduced by a few additional oracle calls, where each call may either return a smaller core, or a stable model possibly improving the current overestimate.
Within this respect, we implemented two strategies, referred to as \emph{linear} and \emph{reiterated progression based shrinking} (Section~\ref{sec:min}).
Concerning the second question, the overhead introduced by the additional oracle calls is often mitigated by the performance gain obtained thanks to the smaller unsatisfiable cores that the algorithm has to analyze.
Indeed, we provide empirical evidence that often the running time of our core based algorithm sensibly decreases when core shrinking is performed (Section~\ref{sec:experiment}).
The advantage of introducing our strategy for core shrinking is also confirmed by a comparison with \textsc{clasp} \cite{DBLP:conf/lpnmr/GebserKK0S15}:
even if our solver, \textsc{wasp} \cite{DBLP:conf/lpnmr/AlvianoDLR15}, is in general slower than \textsc{clasp} at completing stable model searches, its performance is sufficiently improved by core shrinking that the two solvers are almost on par in terms of solved instances, with the crucial difference that \textsc{wasp} provides both overestimates and underestimates during the computation, while ones or the others are produced by \textsc{clasp} only after running to completion.

\section{Background}

Let \A be a set of (propositional) \emph{atoms} comprising $\bot$.
A \emph{literal} is an atom $p$ possibly preceded by zero or more occurrences of the \emph{default negation} symbol $\naf$.
A \emph{sum} aggregate is of the form:
\begin{equation}\label{eq:sum}
    \SUM[w_1:\ell_1, \ldots, w_n:\ell_n] \odot k
\end{equation}
where $n, k, w_1,\ldots,w_n$ are nonnegative integers, $\ell_1, \ldots, \ell_n$ are literals, and $\odot \in \{<,\leq,\geq,>,=,\neq\}$.
If $w_1 = \cdots = w_n = 1$, (\ref{eq:sum}) is also called \emph{count} aggregate, and denoted as $\COUNT[\ell_1, \ldots, \ell_n] \odot k$.
A \emph{rule} $r$ is an implication $H(r) \leftarrow B(r)$, where $H(r)$ is a disjunction of atoms, and $B(r)$ is a conjunction of literals and aggregates.
$H(r)$ and $B(r)$ are called \emph{head} and \emph{body} of $r$, and abusing of notation also denote the sets of their elements.
If $H(r) \subseteq \{\bot\}$, then $r$ is called \emph{integrity constraint}.
A \emph{program} $\Pi$ is a set of rules.
Let $\mathit{At}(\Pi)$ denote the set of atoms occurring in $\Pi$.

An \emph{interpretation} $I$ is a set of atoms not containing $\bot$.
Relation $\models$ is inductively defined as follows:
for $p \in \A$, $I \models p$ if $p \in I$;
$I \models \naf \ell$ if $I \not\models \ell$;
for $A$ of the form (\ref{eq:sum}), $I \models A$ if $\sum_{i \in [1..n],\ I \models \ell_i}{w_i} \odot k$;
for a rule $r$, $I \models B(r)$ if $I \models \ell$ for all $\ell \in B(r)$, and $I \models r$ if $I \cap H(r) \neq \emptyset$ whenever $I \models B(r)$;
for a program $\Pi$, $I \models \Pi$ if $I \models r$ for all $r \in \Pi$.
$I$ is a \emph{model} of a literal, aggregate, rule, or program $\pi$ if $I \models \pi$.
Note that $I \not\models \bot$, and $I \models \naf\bot$, for any interpretation $I$;
let $\top$ be a shortcut for $\naf\bot$.

The \emph{reduct} $\Pi^I$ of a program $\Pi$ with respect to an interpretation $I$ is obtained from $\Pi$ as follows:
(i) any rule $r$ such that $I \not\models B(r)$ is removed;
(ii) any negated literal $\ell$ such that $I \not\models \ell$ is replaced by $\bot$;
(iii) any negated literal $\ell$ such that $I \models \ell$ is replaced by $\top$.
An interpretation $I$ is a \emph{stable model} of a program $\Pi$ if $I \models \Pi$, and there is no $J \subset I$ such that $J \models \Pi^I$.
Let $\SM(\Pi)$ denote the set of stable models of $\Pi$.
A program $\Pi$ is \emph{coherent} if $\SM(\Pi) \neq \emptyset$; otherwise, $\Pi$ is \emph{incoherent}.

\begin{example}\label{ex:sm}
Let $\Pi_1$ be the program comprising the following five rules:
\begin{equation*}
    \begin{array}{lllll}
        a \vee c \leftarrow \naf b, \naf d & a \leftarrow \naf b, c & c \leftarrow a, b & b \leftarrow a, c & d \leftarrow \naf\naf d
    \end{array}
\end{equation*}
Its stable models are $I_1 = \{a\}$ and $I_2 = \{d\}$, and the associated program reducts are the following:
$\Pi_1^{I_1} = \{a \vee c \leftarrow\}$ and $\Pi_1^{I_2} = \{d \leftarrow\}$, where $\top$ in rule bodies is omitted for simplicity.
\hfill$\blacksquare$
\end{example}

A \emph{weak constraint} $r$ is of the form $w@l \cla B(r)$, where $B(r)$ is a conjunction of literals and aggregates, and $w,l$ are positive integers
denoted respectively $\mathit{weight}(r)$ and $\mathit{level}(r)$.
For a multiset \W of weak constraints, and $l \geq 1$, let $\W^l$ denote the multiset of weak constraints in \W whose level is $l$.
The \emph{cost} of an interpretation $I$ for $\W^l$ is $\W^l(I) := \sum_{r \in \W^l,\ I \models B(r)}{\mathit{weight(r)}}$.
For any pair $I,J$ of interpretations, $J <_{\W} I$ if there is $l \geq 1$ such that both $\W^l(J) < \W^l(I)$, and $\W^{l+i}(J) \leq \W^{l+i}(I)$ for all $i \geq 1$.
$I$ is an \emph{optimum stable model} of a program $\Pi$ with respect to a multiset \W of weak constraints if $I \in \SM(\Pi)$, and there is no $J \in \SM(\Pi)$ such that $J <_\W I$.
Let $\OSM(\Pi,\W)$ denote the set of optimum stable models of $\Pi$ with respect to \W.

\begin{example}\label{ex:osm}
Continuing with Example~\ref{ex:sm}, let $\W_1$ be 
$\{1@2 \cla d$, $2@1 \cla a$, $2@1 \cla b$, $1@1 \cla c\}$.
Hence, $\OSM(\Pi_1,\W_1)$ only contains $I_1$, with $\W_1^2(I_1) = 0$ and $\W_1^1(I_1) = 2$.
Indeed, $I_2$ is such that \linebreak[5] \mbox{$\W_1^2(I_2) = 1$} and $\W_1^1(I_1) = 0$, and therefore it is discarded because $I_1 <_{\W_1} I_2$.
\hfill$\blacksquare$
\end{example}

\section{Optimum stable model search}

The computational problem analyzed in this paper is referred to as \emph{optimum stable model search}:
Given a (coherent) program $\Pi$ and a multiset of weak constraints \W, compute an optimum stable model $I^* \in \OSM(\Pi,\W)$.
Currently available algorithms for this problem are either inefficient, as linear search sat-unsat (Section~\ref{sec:linear}), or not anytime, as algorithms based on unsatisfiable core analysis;
\textsc{one} is among the core-based algorithms originally introduced for MaxSAT \cite{DBLP:conf/ijcai/AlvianoDR15}, and is adapted to ASP in Section~\ref{sec:one}.
A concrete strategy to obtain anytime variants of core-based algorithms is then presented in Section~\ref{sec:min};
it consists in a progression (or linear) search to shrink unsatisfiable cores that may also discover better and better stable models.

\subsection{Inefficiencies in linear search sat-unsat}\label{sec:linear}

An immediate algorithm for addressing this problem is known as \emph{linear search sat-unsat}:
a first stable model is searched, and new stable models of improved cost are iteratively computed until an incoherence arises, proving optimality of the last computed model.
To ease the definition of such an algorithm, weak constraints are transformed into \emph{relaxed} integrity constraints by introducing fresh literals called \emph{soft literals}:
whenever the body of a weak constraint is true, the associated soft literal is inferred false so to satisfy the relaxed integrity constraint.
Formally, for a weak constraint $r$, its relaxed form is $\bot \leftarrow B(r) \wedge s_r$, where $s_r$ is the soft literal associated with $r$, and defined by a \emph{choice rule} $s_r \leftarrow \naf\naf s_r$.
For a multiset of weak constraints \W, let $\mathit{relax}(\W)$ be the set of relaxed integrity constraints and choice rules obtained from the weak constraints in \W, and $\mathit{soft}(\W)$ be the set of associated soft literals.

\begin{example}\label{ex:relax}
Continuing with Example~\ref{ex:osm}, $\mathit{relax}(\W_1^2)$ is $\{\bot \leftarrow d, s_1\} \cup \{s_1 \leftarrow \naf\naf s_1\}$, where $s_1$ is the soft literal in $\mathit{soft}(\W_1^2)$, and $\mathit{relax}(\W_1^1)$ is $\{\bot \leftarrow a, s_2;$ $\bot \leftarrow b, s_3;$ $\bot \leftarrow c, s_4\} \cup \{s_i \leftarrow \naf\naf s_i \mid i \in [2..4]\}$, where $s_2,s_3,s_4$ are the soft literals in $\mathit{soft}(\W_1^1)$.
\hfill$\blacksquare$
\end{example}

\begin{algorithm}[t]
    \SetKwInOut{Input}{Input}
    \SetKwInOut{Output}{Output}
    \Input{A coherent program $\Pi$, and a nonempty multiset of weak constraints $\W$.}
    \Output{An optimum stable model $I^* \in \OSM(\Pi,\W)$.}
    $V := \mathit{At}(\Pi)$\tcp*{visible atoms}
    \While{$\W \neq \emptyset$}{
        $l := \max\{i \in \mathbb{N}^+ \mid \W^i \neq \emptyset\}$;\quad
        $\Pi := \Pi \cup \mathit{relax}(\W^l)$;\quad
        $\mathit{ub} := 1+\sum_{r \in \W^l}{\mathit{weight}(r)}$\;
        \lIf{$I^*$ \textbf{is} defined}{$\mathit{ub} := \W^l(I^*)$}
        $(\mathit{res},I) := \mathit{solve}(\Pi \cup \{\bot \leftarrow \SUM[(\mathit{weight}(r) : \naf s_r) \mid s_r \in \mathit{soft}(\W^l)] \geq \mathit{ub}\})$\;\label{alg:linear:ln:solve}
        \lIf{$\mathit{res}$ \textbf{is} $\mathit{COHERENT}$}{
            $I^* := I \cap V$;\quad
            $\mathit{ub} := \W^l(I)$;\quad
            \textbf{goto}~\ref{alg:linear:ln:solve}%
        }
        $\Pi := \Pi \cup \{\bot \leftarrow \SUM[(\mathit{weight}(r) : \naf s_r) \mid s_r \in \mathit{soft}(\W^l)] \neq \mathit{ub}\}$;\quad
        $\W := \W \setminus \W^l$\;
    }
    \Return $I^*$\;
    \caption{Linear search sat-unsat}\label{alg:linear}
\end{algorithm}

Linear search sat-unsat is reported as Algorithm~\ref{alg:linear}.
Levels are processed one at a time, from the greatest to the smallest, by relaxing weak constraints (line~3) and iteratively searching for stable models improving the current upper bound (lines~5--6).
When the upper bound cannot be further improved, an integrity constraint is added to the program in order to discard all stable models of cost different from the upper bound (line~7), so that the next level can be processed correctly.

\begin{example}\label{ex:linear}
Continuing with the previous examples, Algorithm~\ref{alg:linear} with input $\Pi_1$ and $\W_1$ would relax weak constraints in $\W_1^2$, and set $\mathit{ub}$ to 2.
After that, a stable model of the following program:
\begin{equation*}
    \begin{array}{lllll}
        a \vee c \leftarrow \naf b, \naf d & a \leftarrow \naf b, c & c \leftarrow a, b & b \leftarrow a, c & d \leftarrow \naf\naf d \\
        \bot \leftarrow d, s_1 & 
        \multicolumn{2}{l}{s_1 \leftarrow \naf\naf s_1} &
        \multicolumn{2}{l}{\bot \leftarrow \SUM[1 : \naf s_1] \geq 2}
    \end{array}
\end{equation*}
is searched, and $I_2 = \{d\}$ may be found.
In this case, $\mathit{ub}$ is decreased to 1, and a stable model satisfying the new constraint $\bot \leftarrow \SUM[1 : \naf s_1] \geq 1$ is searched.
Eventually, $I_1' = \{a,s_1\}$ is found, and $\mathit{ub}$ is set to 0.
The next stable model search trivially fails because of $\bot \leftarrow \SUM[1 : \naf s_1] \geq 0$, and therefore the cost for level 2 is fixed to 0 by adding $\bot \leftarrow \SUM[1 : \naf s_1] \neq 0$.
Weak constraints of level 1 are relaxed, and $\mathit{ub}$ is set to $\W_1^1(I_1) = 2$.
A stable model of the following program:
\begin{equation*}
    \begin{array}{lllll}
        a \vee c \leftarrow \naf b, \naf d & a \leftarrow \naf b, c & c \leftarrow a, b & b \leftarrow a, c & d \leftarrow \naf\naf d \\
        \bot \leftarrow d, s_1 & 
        \multicolumn{2}{l}{s_i \leftarrow \naf\naf s_i \quad (\forall i \in [1..4])} &
        \multicolumn{2}{l}{\bot \leftarrow \SUM[1 : \naf s_1] \neq 0}\\
        \bot \leftarrow a, s_2 &
        \bot \leftarrow b, s_3 &
        \bot \leftarrow c, s_4 &
        \multicolumn{2}{l}{\bot \leftarrow \SUM[2 : \naf s_2, 2 : \naf s_3, 1 : \naf s_4] \geq 2}
    \end{array}
\end{equation*}
is searched, but none is found.
The algorithm terminates returning $I_1' \cap \{a,b,c,d\} = \{a\} = I_1$.
\hfill$\blacksquare$
\end{example}

The main drawback of linear search sat-unsat is that in practice it is quite inefficient.
In fact, the sum aggregate introduced in order to improve the current upper bound (line~4 of Algorithm~\ref{alg:linear}) allows the solver to discard several stable models, but often does not provide sufficiently strong evidences to help the refutation process.
It is because of this fact that linear search sat-unsat is often unable to prove optimality of a stable model, even if it is provided externally.
In fact, an integrity constraint of the form $\bot \leftarrow \SUM[(w_i : \ell_i) \mid 1 \leq i \leq n] \geq k$ is satisfied by exponentially many interpretations.
Even when all weights are 1, the number of satisfying interpretations is in the order of $\binom{n}{k-1}$.
It is a prohibitive number even for small values of $k$ when $n$ is a large integer, as it usually happens for Algorithm~\ref{alg:linear}, where $n$ is essentially the number of weak constraints.

\subsection{Almost silent unsatisfiable core analysis}\label{sec:one}

The idea underlying algorithms based on unsatisfiable core analysis is the following:
a set of soft literals that cannot be jointly satisfied is identified, and the reason of unsatisfiability is removed by adjusting the program and the soft literals;
the process is repeated until a stable model is found, which is also guaranteed to be optimum.
To ease the definition of such an algorithm, the notion of unsatisfiable core is given in terms of \emph{assumptions}:
modern ASP solvers accept as input a set $S$ of atoms, called assumptions, in addition to the usual logic program $\Pi$, and return a stable model $I$ of $\Pi$ such that $S \subseteq I$, if it exists;
otherwise, they return a set $C \subseteq S$ such that $\Pi \cup \{\bot \leftarrow \naf p \mid p \in C\}$ is incoherent, which is called \emph{unsatisfiable core}.

\pagebreak[4]
\begin{example}\label{ex:core}
Consider program $\Pi_1$ from Example~\ref{ex:sm} with the addition of $\mathit{relax}(\W_1^2)$ and $\mathit{relax}(\W_1^1)$:
\begin{equation*}
    \begin{array}{lllllll}
        a \vee c \leftarrow \naf b, \naf d & a \leftarrow \naf b, c & c \leftarrow a, b & b \leftarrow a, c & d \leftarrow \naf\naf d \\
        \bot \leftarrow d, s_1 & 
        \bot \leftarrow a, s_2 &
        \bot \leftarrow b, s_3 &
        \bot \leftarrow c, s_4 &
        s_i \leftarrow \naf\naf s_i \quad (\forall i \in [1..4])
    \end{array}
\end{equation*}
If $\{s_1,\ldots,s_4\}$ is the set of assumptions, the unsatisfiable cores are $\{s_1,s_2\}$, and its supersets.
\hfill$\blacksquare$
\end{example}

\begin{algorithm}[t]
    \SetKwInOut{Input}{Input}
    \SetKwInOut{Output}{Output}
    \Input{A coherent program $\Pi$, and a nonempty multiset of weak constraints $\W$.}
    \Output{An optimum stable model $I^* \in \OSM(\Pi,\W)$.}
    $V := \mathit{At}(\Pi)$;\qquad \lFor{$p \in \A$}{$w(p) := 0$}
    \While{$\W \neq \emptyset$}{
        $l := \max\{i \in \mathbb{N}^+ \mid \W^i \neq \emptyset\}$;\quad
        \!$\Pi := \Pi \cup \mathit{relax}(\W^l)$;\quad
        \!$\mathit{lb} := 0$;\quad
        \!$\mathit{ub} := \infty$;\quad
        \!$\mathit{stratum} := \infty$\;
        \lFor(\tcp*[f]{soft literals initialization}){$s_r \in \mathit{soft}(\W^l)$}{$w(s_r) := \mathit{weight}(r)$}
        \lIf{$I^*$ \textbf{is} defined}{$\mathit{ub} := \W^l(I^*)$;\quad \textsc{hardening}($\Pi,w,\mathit{lb},\mathit{ub}$)}
        $\mathit{stratum} := \max_{p \in \A, w(p) < \mathit{stratum}}{w(p)}$\label{alg:one:ln:stratum}\tcp*{next stratum}
        $(\mathit{res},I,C) := \mathit{solve}(\Pi, \{p \in \A \mid w(p) \geq \mathit{stratum}\})$\;\label{alg:one:ln:solve}
        \If{$\mathit{res}$ \textbf{is} $\mathit{INCOHERENT}$}{
            \setcounter{AlgoLine}{15}
            Let $C$ be $\{p_0,\ldots,p_n\}$ (for some $n \geq 0$), and $s_1,\ldots,s_n$ be fresh atoms\;\label{alg:one:ln:analyze}
            $\mathit{lb} := \mathit{lb} + \mathit{stratum}$\;
            \lFor(\tcp*[f]{remainders}){$i \in [0..n]$}{$w(p_i) := w(p_i) - \mathit{stratum}$}
            \lFor(\tcp*[f]{new soft literals}){$i \in [1..n]$}{$w(s_i) := \mathit{stratum}$}
            $\Pi := \Pi \cup \{s_i \leftarrow \naf\naf s_i \mid i \in [1..n]\} \cup \{\bot \leftarrow s_i, \naf s_{i+1} \mid i \in [1..n-1]\}$
                $\phantom{lxxxxx}\cup \{\bot \leftarrow \COUNT[p_0,\ldots,p_n, \naf s_1,\ldots,\naf s_n] < n\}$%
            \tcp*{relax core}
            \textsc{hardening}($\Pi,w,\mathit{lb},\mathit{ub}$);\quad \textbf{goto}~\ref{alg:one:ln:solve}\;
        }
        
        \lIf{$\W^l(I) < \mathit{ub}$}{
            $I^* := I \cap V$;\quad
            $\mathit{ub} := \W^l(I)$;\quad
            \textsc{hardening}($\Pi,w,\mathit{lb},\mathit{ub}$)%
        }
        
        \lIf{$\exists p \in \A$ \textbf{such that} $1 \leq w(p) < \mathit{stratum}$}{\textbf{goto}~\ref{alg:one:ln:stratum}}
        $\W := \W \setminus \W^l$\;\label{alg:one:ln:remove}
    }
    \Return $I^*$\;
    \caption{Unsatisfiable core analysis with \textsc{one}}\label{alg:one}
\end{algorithm}
\begin{procedure}[t]
    \lFor{$p \in \A$ \textbf{such that} $\mathit{lb} + w(p) > \mathit{ub}$}{$\Pi := \Pi \cup \{\bot\!\leftarrow\!\naf p\}$; \quad $w(s_r)\!:=\!0$}
    \caption{hardening($\Pi, w, \mathit{lb}, \mathit{ub}$)}
\end{procedure}

The algorithm presented in this paper is \textsc{one}, reported as Algorithm~\ref{alg:one} (lines 9--15 will be \emph{injected} later to shrink unsatisfiable cores).
Every atom $p$ is associated with a weight $w(p)$, initially set to zero (line~1), and levels are processed from the greatest to the smallest by relaxing constraints (line~3).
Note that soft literals are associated with nonzero weights (line~4), and are processed per stratum, i.e., greatest weights are processed first (line~6).
A stable model containing all soft literals in the current stratum is then searched (line~7).
If an unsatisfiable core $\{p_0,\ldots,p_n\}$ is returned (lines~8 and 16), since at least one of $p_0,\ldots,p_n$ must be false in any optimum stable model, the lower bound is increased by the minimum weight among $w(p_0),\ldots,w(p_n)$ (or equivalently by $\mathit{stratum}$; line~17).
Such a quantity is removed from $p_0,\ldots,p_n$ (line~18), and assigned to $n$ new soft literals $s_1,\ldots,s_n$ (line~19).
The new soft literals, and the following constraint (line~20):
\begin{equation*}
    \bot \leftarrow \COUNT[p_0,\ldots,p_n, \naf s_1,\ldots,\naf s_n] < n
\end{equation*}
enforce the next call to function $\mathit{solve}$ to search for a stable model satisfying at least $n$ literals among $p_0,\ldots,p_n$.
Moreover, note that symmetry breakers of the form $\bot \leftarrow s_i, \naf s_{i+1}$ are also added to $\Pi$, so that $s_i$ is true if and only if at least $n-i+1$ literals among $p_0,\ldots,p_n$ are true.

The current stratum is then processed again (line~21), until a stable model is found.
In this case, the upper bound is possibly improved (line~22), and the stratum is extended to soft literals of smaller weight (line~23), if there are.
Otherwise, the processed stratum covers all soft literals, and weak constraints of the current level are removed so that the next level can be considered (line~24).
Note that the algorithm also includes the \textsc{hardening} procedure, which essentially enforces truth of soft literals that are guaranteed to belong to all optimum stable models.

\begin{example}\label{ex:one}
Initially, the weak constraint in $\W_1^2$ is relaxed, and the soft literal $s_1$ is associated with weight $w(s_1) = 1$.
A stable model for the assumption $\{s_1\}$ and the following program:
\begin{equation*}
    \begin{array}{lllllll}
        a \vee c \leftarrow \naf b, \naf d & a \leftarrow \naf b, c & c \leftarrow a, b & b \leftarrow a, c & d \leftarrow \naf\naf d \\
        \bot \leftarrow d, s_1 & 
        s_1 \leftarrow \naf\naf s_1
    \end{array}
\end{equation*}
is searched, and $I_1' = \{a,s_1\}$ is found.
Hence, $\mathit{ub}$ is set to 0, and $s_1$ is hardened by adding $\bot \leftarrow \naf s_1$.
Weak constraints in $\W_1^1$ are relaxed, $w$ is such that $w(s_2) = 2$, $w(s_3) = 2$ and $w(s_4) = 1$, and $\mathit{ub}$ is set to $\W_1^1(I_1) = 2$.
A stable model for the assumptions $\{s_2,s_3\}$ and the following program:
\begin{equation*}
    \begin{array}{lllllll}
        a \vee c \leftarrow \naf b, \naf d & a \leftarrow \naf b, c & c \leftarrow a, b & b \leftarrow a, c & d \leftarrow \naf\naf d & \bot \leftarrow \naf s_1 \\
        \bot \leftarrow d, s_1 & 
        \bot \leftarrow a, s_2 &
        \bot \leftarrow b, s_3 &
        \bot \leftarrow c, s_4 &
        \multicolumn{2}{l}{s_i \leftarrow \naf\naf s_i \quad (\forall i \in [1..4])}
    \end{array}
\end{equation*}
is searched, and an unsatisfiable core, say $\{s_2,s_3\}$, is returned.
The lower bound $\mathit{lb}$ is set to 2, the program is extended with the following rules:
\begin{equation*}
    \begin{array}{llllll}
        s_5 \leftarrow \naf\naf s_5 &
        \bot \leftarrow \COUNT[s_2, s_3, \naf s_5] < 1
    \end{array}
\end{equation*}
and function $w$ is now such that $w(s_2) = 0$, $w(s_3) = 0$, $w(s_4) = 1$ and $w(s_5) = 2$.
Since $\mathit{lb} = \mathit{ub}$, all soft literals are hardened.
The algorithm terminates returning $I_1' \cap \{a,b,c,d\} = \{a\} = I_1$.
\hfill$\blacksquare$
\end{example}

The algorithm described in this section is almost silent, as it essentially runs to completion without printing any suboptimal stable models but those computed when changing stratum or level.
This fact is actually limiting the use of core-based algorithms in ASP, as common instances usually have one or few levels, and also weights are often uniform.
In particular, when all weak constraints have the same level and the same weight, the algorithm is completely silent.

\begin{example}\label{ex:dumb}
Consider now a multiset $\W_2$ of weak constraints obtained from $\W_1$ by setting all weights and levels to 1.
Algorithm~\ref{alg:one} starts with assumptions $\{s_1,\ldots,s_4\}$ and the following program:
\begin{equation*}
    \begin{array}{lllllll}
        a \vee c \leftarrow \naf b, \naf d & a \leftarrow \naf b, c & c \leftarrow a, b & b \leftarrow a, c & d \leftarrow \naf\naf d \\
        \bot \leftarrow d, s_1 & 
        \bot \leftarrow a, s_2 &
        \bot \leftarrow b, s_3 &
        \bot \leftarrow c, s_4 &
        \multicolumn{2}{l}{s_i \leftarrow \naf\naf s_i \quad (\forall i \in [1..4])}
    \end{array}
\end{equation*}
Hence, an unsatisfiable core, say $\{s_1,\ldots,s_4\}$, is returned, and the following rules are added:
\begin{equation*}
    \begin{array}{llllll}
        s_i \leftarrow \naf\naf s_i \quad (\forall i \in [5..7]) &
        \bot \leftarrow s_5, \naf s_6 &
        \bot \leftarrow s_6, \naf s_7 \\
        \multicolumn{3}{l}{\bot \leftarrow \COUNT[s_1, s_2, s_3,s_4, \naf s_5, \naf s_6, \naf s_7] < 3}
    \end{array}
\end{equation*}
where $s_5,s_6,s_7$ are the new soft literals.
The assumptions are now $\{s_5,s_6,s_7\}$, and either $I_1'' = \{a,s_1,s_3,s_4,s_5,s_6,s_7\}$ or $I_2' = \{d,s_2,s_3,s_4,s_5,s_6,s_7\}$ is found.
The algorithm terminates because $\mathit{lb} = \mathit{ub} = \W_2^1(I_1'') = \W_2^1(I_2') = 1$.
Note that no stable model was found before running to completion, and in general Algorithm~\ref{alg:one} has to analyze several unsatisfiable cores to terminate.
\hfill$\blacksquare$
\end{example}

\subsection{Unsatisfiable core shrinking and upper bounds}\label{sec:min}

\begin{algorithm}[t]
\setcounter{AlgoLine}{8}
    $m := -1$;\quad
    $\mathit{pr} := 1$\;
    Let $C$ be $\{p_0,\ldots,p_n\}$ (for some $n \geq 0$)\;\label{alg:min:ln:core}
    $(\mathit{res},I,C') := \mathit{solve\_with\_budget}(\Pi, \{p_i \mid i \in [0..m+\mathit{pr}]\})$\;
    \lIf(\tcp*[f]{smaller core found}){$\mathit{res}$ \textbf{is} $\mathit{INCOHERENT}$}{$C := C'$}
    \lIf{$\mathit{res}$ \textbf{is} $\mathit{COHERENT}$ {\bf and} $\W^l(I) < \mathit{ub}$}{$I^* := I \cap V$;\quad $\mathit{ub} := \W^l(I)$}

    \lIf(\tcp*[f]{reiterate progression}){$m + 2 \cdot \mathit{pr} \geq |C|-1$}{$m := m + \mathit{pr}$;\quad $\mathit{pr} := 1/2$}
    \lIf(\tcp*[f]{increase progression}){$m + 2 \cdot \mathit{pr} < |C|-1$}{$\mathit{pr} := 2 \cdot \mathit{pr}$;\quad\textbf{goto}~\ref{alg:min:ln:core}}
    \caption{Unsatisfiable core shrinking with reiterated progression}\label{alg:min}
\end{algorithm}

Unsatisfiable cores returned by function $\mathit{solve}$ are not subset minimal in general.
The non-minimality of the unsatisfiable core is justified both theoretically and practically:
linearly many coherence checks are required in general to verify the minimality of an unsatisfiable core, hence giving a $\Delta^P_3$-complete problem;
on the other hand, extracting an unsatisfiable core after a stable model search failure is quite easy and usually implemented by identifying the assumptions involved in the refutation.
The non minimality of the analyzed unsatisfiable cores may affect negatively the performance of subsequent calls to function $\mathit{solve}$ due to aggregation over large sets.
However, it also gives an opportunity to improve Algorithm~\ref{alg:one}:
the size of unsatisfiable cores can be reduced by performing a few stable model searches within a given budget on the running time.
In more detail, Algorithm~\ref{alg:min} is injected between lines~8 and 16 of Algorithm~\ref{alg:one}.
It implements a progression search in the unsatisfiable core $\{p_0,\ldots,p_n\}$:
the size of the assumptions passed to function $\mathit{solve\_with\_budget}$ is doubled at each call (line~15), and the progression is reiterated when all assumptions are covered (line~14).
If function $\mathit{solve\_with\_budget}$ terminates within the given budget, it either returns a smaller unsatisfiable core (line~12), or a stable model that possibly improves the current upper bound (line~13).

\begin{example}\label{ex:min}
Consider $\W_2$ from Example~\ref{ex:dumb}, and the unsatisfiable core $\{s_1,\ldots,s_4\}$ returned after the first call to function $\mathit{solve}$.
The shrinking process searches a stable model with assumption $\{s_1\}$, and $\{a,s_1\} \cup X$ (for some $X \subseteq \{s_3,s_4\}$) may be found within the allotted budget.
In any case, a stable model satisfying the assumptions $\{s_1,s_2\}$ is searched, and the unsatisfiable core $\{s_1,s_2\}$ may be returned if the budget is sufficient.
Otherwise, the progression is reiterated, and one more soft literal is added to the assumptions.
Hence, $\{s_1,s_2,s_3\}$ may be returned as an unsatisfiable core, or the budget may be insufficient and the original unsatisfiable core will be processed.
\hfill$\blacksquare$
\end{example}

As an alternative, the shrinking procedure reported in Algorithm~\ref{alg:min} can be modified as follows:
variable $\mathit{pr}$ is not doubled in line~15, but instead it is incremented by one, i.e., $\mathit{pr} := \mathit{pr} + 1$.
The resulting procedure is called \emph{linear based shrinking}.
For unsatisfiable cores of size 4 or smaller, as those considered in Example~\ref{ex:min}, the two shrinking procedures coincide, while in general linear based shrinking performs more stable model searches.

\begin{theorem}
Let $C$ be an unsatisfiable core.
Function $\mathit{solve\_with\_budget}$ is invoked $O((\log{|C|})^2)$ times by  progression based shrinking, and $O(|C|)$ times by linear based shrinking.
\end{theorem}
\begin{proof}
The worst case occurs when $\mathit{solve\_with\_budget}$ never returns $\mathit{COHERENT}$, either because the tested set of assumptions is not an unsatisfiable core, or because the allotted budget is insufficient.
This is the case, for example, if $C$ is already minimal.
For progression based shrinking, line~14 of Algorithm~\ref{alg:min} is first executed after $k := \ceil{\log{|C|}}$ executions of lines~11 and 15.
After that, and in the worst case, the process is repeated on half of the literals in $C$, for $k-1$ executions of lines~11 and 15.
Hence, for a total of $k(k+1)/2$ executions of function $\mathit{solve\_with\_budget}$.
For linear based shrinking, instead, line~14 of Algorithm~\ref{alg:min} is first executed after $|C|-1$ executions of lines~11 and 15.
After that, the process terminates because $m \geq |C|-1$.
\end{proof}

\section{Implementation and experiments}\label{sec:experiment}

Algorithm \textsc{one} \cite{DBLP:conf/ijcai/AlvianoDR15} has been implemented in \textsc{wasp} \cite{DBLP:conf/lpnmr/AlvianoDLR15}, an ASP solver also supporting, among other algorithms, linear search sat-unsat.
The implementation of \textsc{one} optionally includes the two shrinking procedures described in Section~\ref{sec:min}, so that both underestimates and overestimates can be produced by \textsc{wasp} in any case, weighted or unweighted.
Currently, the time budget of function $\mathit{solve\_with\_budget}$ is fixed to 10 seconds, but the architecture of \textsc{wasp} can easily accommodate alternative options, such as a budget proportional to the time required to find the unsatisfiable core to be shrank.

\textsc{wasp} also implements \emph{disjoint cores analysis}, which is essentially a preliminary step where only soft literals associated with weak constraints in the input are passed as assumptions to function $\mathit{solve}$, while new soft literals introduced by the analysis of detected cores are temporarily ignored.
Disjoint cores analysis terminates with the detection of a stable model, and after that algorithm \textsc{one} is run as usual not distinguishing between initial and new soft literals.

In order to assess empirically the impact of these shrinking procedures to the performance of \textsc{wasp}, instances (with weak constraints) from the ASP Competition 2015 \cite{DBLP:conf/lpnmr/GebserMR15} were considered.
Moreover, \textsc{wasp} was also compared with \textsc{clasp} \cite{DBLP:conf/lpnmr/GebserKK0S15}, which implements linear search sat-unsat (strategy \texttt{bb,1}) and \textsc{oll} (strategy \texttt{usc}; \texttt{usc,1} with disjoint cores analysis; \citeNP{DBLP:conf/iclp/AndresKMS12}), a core based algorithm that inspired the definition of \textsc{one}.
The experiments were run on an Intel Xeon 2.4 GHz with 16 GB of memory, and time and memory were limited to 20 minutes and 15 GB, respectively.

\begin{table}[b!]
	\footnotesize
	\caption{
		Solved instances, and wins in parenthesis (180 testcases; 20 for each problem).
	}\label{tab:experiments}
	\tabcolsep=0.030cm	
	\centering
	\begin{tabular}{lrrrrrrrrrrr}
		\toprule
		& \multicolumn{3}{c}{\clasp}	&\multicolumn{4}{c}{\wasp} & \multicolumn{3}{c}{\waspdisj} \\
		\cmidrule{2-4}\cmidrule{5-8}\cmidrule{9-11}
	Problem 	&	 \textsc{linSU} 	&	 \textsc{oll} 	&	 \textsc{oll}+\textsc{disj} 	&	 \textsc{linSU} 	&	 \textsc{one} 	&	 \textsc{Lshr} 	&	\textsc{Pshr} 	&	 \textsc{one} 	&	 \textsc{Lshr} 	&	 \textsc{Pshr} \\
 \cmidrule{1-11}																						
 ADF   									&	 20 (20) 	&	 19 (19) 	&	 19 (19) 	&	 18 (18)  	&	 14 (14) 	&	 17 (17) 	&	18 (18) 	&	 14 (14) 	&	 18 (18)	&	 18 (18)  \\
 ConnectedStillLife     	&	 7 (10) 	&	 11 (11) 	&	 13 (13) 	&	 0    \phantom{0}(0) 	&	 13 (13) 	&	 10 (10) 	&	10 (12) 	&	 13 (13) 	&	 9 (11)	&	 10 (10)  \\
 CrossingMin    				&	 7   \phantom{0}(7) 	&	 19 (19) 	&	 19 (19) 	&	 2    \phantom{0}(2) 	&	 19 (19) 	&	 19 (19) 	&	19 (19) 	&	 19 (19) 	&	 19 (20)	&	 19 (19)  \\
 MaxClique 						&	 0    \phantom{0}(4)  	&	 16 (16) 	&	 16 (16) 	&	 0    \phantom{0}(0)  	&	 10 (10) 	&	 14 (14) 	&	16 (16) 	&	 11 (11) 	&	 15 (15)	&	 16 (16)  \\
 MaxSAT   							&	 7    \phantom{0}(7) 	&	 15 (15) 	&	 15 (15) 	&	 5    \phantom{0}(5) 	&	 19 (19) 	&	20 (20) 	&	20 (20) 	&	 18 (18) 	&	 20 (20)	&	 20 (20)  \\
 SteinerTree  					&	 3 (16)  	&	 1 \phantom{0}(1) 	&	 1    \phantom{0}(1) 	&	 1    \phantom{0}(1) 	&	 1    \phantom{0}(5) 	&	 1    \phantom{0}(3) 	&	1    \phantom{0}(3) 	&	 0    \phantom{0}(1) 	&	 0    \phantom{0}(0)	&	 0    \phantom{0}(1)  \\
 SystemSynthesis 			&	 0    \phantom{0}(1) 	&	 4    (11)  	&	 4    (11) 	&	 0    \phantom{0}(4) 	&	 0    \phantom{0}(2) 	&	 0    \phantom{0}(2) 	&	0    \phantom{0}(2) 	&	 0    \phantom{0}(2) 	&	 0    \phantom{0}(2)	&	 0    \phantom{0}(2)  \\
 ValvesLocation  				&	 16 (17) 	&	 9    \phantom{0}(9) 	&	 11 (11) 	&	 11 (11) 	&	 15 (17) 	&	 16 (19) 	&	16 (19) 	&	 14 (14) 	&	 16 (16)	&	 16 (16)  \\
 VideoStreaming 			&	 14 (19) 	&	 9    \phantom{0}(9) 	&	 9    \phantom{0}(9) 	&	 8    \phantom{0}(8)  	&	 0    \phantom{0}(0) 	&	 0    \phantom{0}(0) 	&	0    \phantom{0}(0) 	&	 9 (10) 	&	 9    \phantom{0}(9)	&	 9    \phantom{0}(9)  \\
 \cmidrule{1-11}																						
 \textbf{Total} 	&	 74 (101) 	&	 103 (110) 	&	 107 (114) 	&	 45 (49) 	&	91 (99) 	&	 97 (104) 	&	100 (109) 	&	98 (102) 	&	 106 (111)	&	 108 (111)  \\
 \bottomrule
	\end{tabular}
\end{table}

An overview of the obtained results is given in Table~\ref{tab:experiments}, where the number of solved instances is reported.
As a first comment, the fact that \textsc{clasp} is in general faster than \textsc{wasp} to complete stable model searches is confirmed by comparing the performance of the two solvers running linear search sat-unsat (\textsc{linSU} in the table; \textsc{clasp} solves 29 instances more than \textsc{wasp}) or the core based algorithms (difference of 12 instances, reduced to 9 if disjoint cores are computed).
This gap is completely filled by adding the shrinking procedures if disjoint cores are computed, and significantly reduced otherwise:
6 instances with linear based shrinking (\textsc{Lshr} in the table), and 3 instances with reitered progression (\textsc{Pshr} in the table).
It is also important to observe that the performance of \textsc{wasp} is boosted by disjoint cores only for instances of Video Streaming, where there are many strata and few unsatisfiable cores, which is the worst case for stratification.

Table~\ref{tab:experiments} also shows the number of wins of each tested strategy, where a strategy wins on a tested instance if it terminates in the allotted time, or it finds the smallest upper bound and none of the tested strategies terminated.
Within this respect, the linear search sat-unsat algorithm implemented by \textsc{clasp} has a significant performance improvement, in particular for instances of Steiner Tree and Video Streaming.
However, even according to this measure, a better performance is obtained by \textsc{clasp} using \textsc{oll} and disjoint cores analysis, which is also reached by \textsc{wasp} thanks to the core shrinking procedures introduced in this paper.

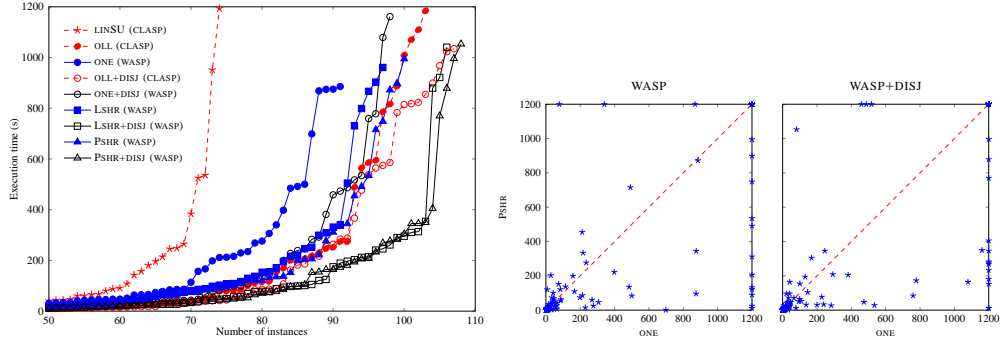
\begin{figure}[t]
	\figrule
	\begin{tikzpicture}[scale=0.6]
	\pgfkeys{%
		/pgf/number format/set thousands separator = {}}
	\begin{axis}[
	scale only axis
	, font=\scriptsize
	, x label style = {at={(axis description cs:0.5,0.04)}}
	, y label style = {at={(axis description cs:0.05,0.5)}}
	, xlabel={Number of instances}
	, ylabel={Execution time (s)}
	, xmin=50, xmax=110
	, ymin=0, ymax=1200
	, legend style={at={(0.18,0.96)},anchor=north, draw=none,fill=none}
	, legend columns=1
	, width=0.7\textwidth
	, height=0.5\textwidth
	, ytick={0,200,400,600,800,1000,1200}
	, xtick={50,60,70,80,90,100,110}
	, major tick length=2pt
	]
	\addplot [mark size=2.5pt, color=red, mark=star, dashed] [unbounded coords=jump] table[col sep=semicolon, y index=1] {./cactus.csv}; 
    	\addlegendentry{\textsc{linSU (clasp)}}

	\addplot [mark size=2pt, color=red, mark=*, dashed] [unbounded coords=jump] table[col sep=semicolon, y index=2] {./cactus.csv}; 
    	\addlegendentry{\textsc{oll (clasp)}}
	\addplot [mark size=2pt, color=blue, mark=*] [unbounded coords=jump] table[col sep=semicolon, y index=5] {./cactus.csv}; 
    	\addlegendentry{\textsc{one (wasp)}}
	\addplot [mark size=2pt, color=red, mark=o, dashed, mark options=solid] [unbounded coords=jump] table[col sep=semicolon, y index=3] {./cactus.csv}; 

    	\addlegendentry{\textsc{oll+disj (clasp)}}
	\addplot [mark size=2pt, color=black, mark=o] [unbounded coords=jump] table[col sep=semicolon, y index=6] {./cactus.csv}; 
    	\addlegendentry{\textsc{one+disj (wasp)}}


	\addplot [mark size=2.0pt, color=blue, mark=square*] [unbounded coords=jump] table[col sep=semicolon, y index=7] {./cactus.csv}; 
    	\addlegendentry{\textsc{Lshr (wasp)}}
	\addplot [mark size=2pt, color=black, mark=square] [unbounded coords=jump] table[col sep=semicolon, y index=8] {./cactus.csv}; 
    	\addlegendentry{\textsc{Lshr+disj (wasp)}}

	\addplot [mark size=2.5pt, color=blue, mark=triangle*] [unbounded coords=jump] table[col sep=semicolon, y index=9] {./cactus.csv}; 
    	\addlegendentry{\textsc{Pshr (wasp)}}
	\addplot [mark size=2.5pt, color=black, mark=triangle] [unbounded coords=jump] table[col sep=semicolon, y index=10] {./cactus.csv}; 
    	\addlegendentry{\textsc{Pshr+disj (wasp)}}
	\end{axis}
	\end{tikzpicture}
	%
    %
	\begin{tikzpicture}[scale=0.58]
	\pgfkeys{%
		/pgf/number format/set thousands separator = {}}
	\begin{axis}[
	scale only axis
	, font=\scriptsize
	, x label style = {at={(axis description cs:0.5,0.04)}}
	, y label style = {at={(axis description cs:0.05,0.5)}}
	, xlabel={\textsc{one}}
	, ylabel={\textsc{Pshr}}
	, width=0.35\textwidth
	, height=0.35\textwidth
	, xmin=0, xmax=1200
	, ymin=0, ymax=1200
	, xtick={0,200,400,600,800,1000,1200}
	, ytick={0,200,400,600,800,1000,1200}
	, major tick length=2pt
	, title={\normalsize{\textsc{wasp}}}
	]
	\addplot [mark size=2pt, only marks, color=blue, mark=star, dashed] [unbounded coords=jump] table[col sep=semicolon, x index=2, y index=3] {./scatter.csv}; 	
	\addplot [color=red, dashed] [unbounded coords=jump] table[col sep=semicolon, x index=1, y index=1] {./cactus.csv}; 
	\end{axis}
	\end{tikzpicture}
	\begin{tikzpicture}[scale=0.58]
	\pgfkeys{%
		/pgf/number format/set thousands separator = {}}
	\begin{axis}[
	scale only axis
	, font=\scriptsize
	, x label style = {at={(axis description cs:0.5,0.04)}}
	, xlabel={\textsc{one}}
	, width=0.35\textwidth
	, height=0.35\textwidth
	, xmin=0, xmax=1200
	, ymin=0, ymax=1200
	, xtick={0,200,400,600,800,1000,1200}
	, ytick={0,200,400,600,800,1000,1200}
	, major tick length=2pt
	, title={\normalsize{\textsc{wasp+disj}}}
	, yticklabels={}
	]
	\addplot [mark size=2pt, only marks, color=blue, mark=star, dashed] [unbounded coords=jump] table[col sep=semicolon, x index=4, y index=5] {./scatter.csv}; 	
	\addplot [color=red, dashed] [unbounded coords=jump] table[col sep=semicolon, x index=1, y index=1] {./cactus.csv}; 
	\end{axis}
	\end{tikzpicture}	
	\caption{Solved instances within a bound on the running time (cactus plot on the left), and instance-by-instance comparison of \textsc{one} with and without \textsc{Pshr} (scatter plots on the right).}\label{fig:cactus}
	\figrule
\end{figure}

Concerning the average execution time of the tested algorithms, a cactus plot is reported in 
Fig.~\ref{fig:cactus} (left; \textsc{wasp} using \textsc{linSU} is not reported because it solves less than 50 instances).
The graph highlights that core based algorithms are faster than linear search sat-unsat in more testcases.
Moreover, and more important, the addition of core shrinking does not add overhead to \textsc{wasp}.
This aspect is even more clear in the scatter plots on the right, where the impact of \textsc{Pshr} is shown instance-by-instance.
It can be observed that the majority of points are below the diagonal, meaning that often core shrinking provides a performance gain, and only in a few cases it introduces overhead.
The main reason for this performance improvement is that shrinking a core often implies that subsequently found unsatisfiable cores are smaller:
The cumulative number of literals in the analyzed cores is reduced by at least $68\%$ when shrinking is performed (excluding Steiner Tree, System Synthesis and Video Streaming, for which \textsc{wasp} found few unsatisfiable cores).
The budget is reached at least once in each problem, and often no more than 2 times, with a peak of 20--25 times on average for instances of Max Clique and Still Life.

Another advantage of unsatisfiable core shrinking is that better and better stable models are possibly discovered while computing an optimum stable model.
In order to measure the impact of our strategies within this respect, let us define the \emph{estimate error} $\epsilon$ of the last stable model produced by Algorithm~\ref{alg:one} as follows:
\begin{equation*}
    \epsilon(\mathit{ub},\mathit{lb}) := \left\{
        \begin{array}{cl}
            \frac{\mathit{ub}-\mathit{lb}}{\mathit{lb}} & \text{if } \mathit{ub} \neq \infty \text{ and } \mathit{lb} \neq 0; \\
            \infty & \text{if } \mathit{ub} = \infty, \text{ or both } \mathit{ub} \neq 0 \text{ and } \mathit{lb} = 0;\\
            0 & \text{if } \mathit{ub} = \mathit{lb} = 0. 
        \end{array}
    \right.
\end{equation*}
Hence, the cost associated with the stable model returned by Algorithm~\ref{alg:one} is at most $\epsilon(\mathit{ub},\mathit{lb})$ times greater than the cost of an optimum stable model.
Such a measure is applicable to testcases comprising weak constraints of the same level, which is the case for all problems in our benchmark but ADF and System Synthesis.

\begin{table}[b!]
	\caption{
		Number of solved instances within a given error estimation (140 testcases).
	}\label{tab:approx}
	\tabcolsep=0.030cm	
	\centering
	\begin{tabular}{rcrrrrrrrrrrrr}
		\toprule
		&&& \multicolumn{3}{c}{\wasp}	& \phantom{00} &\multicolumn{3}{c}{\waspdisj} & \phantom{00} & \multicolumn{3}{c}{\clasp + best $\mathit{lb}$ by \textsc{wasp}} \\
		\cmidrule{4-6}\cmidrule{8-10}\cmidrule{12-14}		
\multicolumn{2}{c}{$\epsilon(\mathit{ub},\mathit{lb})$}		&& \textsc{one}	&	\textsc{Lshr}	&	\textsc{Pshr}	&&	\textsc{one}	&	\textsc{Lshr}	&	\textsc{Pshr} && \phantom{xxx} & \textsc{linSU} & \textsc{oll}+\textsc{disj} \\
		\cmidrule{1-14}	
				$0.00\%$	&	(0)		&&	77		&	80	&	82	&&	84	&	88	&	90 &&& 77 &	84 \\
				$\leq \phantom{00}6.25\%$ & (1/16)	&&	77		&	87	&	88	&&	86	&	95	&	94 &&& 90 &	87 \\
				$\leq \phantom{0}12.50\%$ & (1/8)  &&	77	&	91	&	93	&&	86	&	101	&	97 &&& 96 & 88	\\
				$\leq \phantom{0}25.00\%$ & (1/4)	 &&	77	&	97	&	95	&&	92	&	105	&	103 &&& 99 &	99 \\
				$\leq \phantom{0}50.00\%$ & (1/2)	&&	78	&	97	&	97	&&	102	&	105	&	105 &&& 99 &  101	\\
				$\leq 100.00\%$	&	(1)	&&	78	&	97	&	97	&&	104	&	105	&	105 &&& 107 & 105	\\				
 \bottomrule
\end{tabular}
\end{table}

Table~\ref{tab:approx} reports the number of instances for which \textsc{wasp} produced a stable model within a given error estimate.
In particular, the first row shows the number of instances for which an optimum stable model was computed (error estimate is 0).
The last row, instead, shows the number of instances solved with error estimate bounded by 1, and smaller values for the error estimate are considered in the intermediate rows.
It is interesting to observe that without shrinking and disjoint cores analysis, the number of solved instances only increases by 1, confirming the fact that stratification alone is not sufficient to produce suboptimal stable models of reasonable cost.
Better results are instead obtained by adding core shrinking, which gives an increase of up to 17 instances, many of which with an error estimate bounded by $6.25\%$.
Another interesting observation concerns disjoint cores analysis.
The stable model produced after the analysis of all disjoint cores is already sufficient to obtain an error estimate bounded by $100\%$ for many tested instances.
However, many of these stable models have an error estimate greater than $25\%$.
Also in this case, adding core shrinking leads to better results.

For the sake of completeness, also \textsc{clasp} is included in Table~\ref{tab:approx}.
However, since \textsc{clasp} does not print any lower bound, the best value for $\mathit{lb}$ produced by \textsc{wasp} is combined with the upper bounds given by \textsc{clasp} running \textsc{linSU} and \textsc{oll}.
If an error estimate of $100\%$ is acceptable, then the number of stable models produced by \textsc{clasp} is aligned with \textsc{wasp}, or even better.
However, when the error estimate must be less or equal than $50\%$, the combination of disjoint cores analysis and core shrinking implemented by \textsc{wasp} leads to better results in this benchmark.

\section{Related work}

Weak constraints are the analogous of soft clauses in MaxSAT \cite{DBLP:conf/aaai/ChaIKM97}, which was also considered as target language by some ASP solver \cite{DBLP:conf/kr/LiuJN12}.
The main difference is that weak constraints are associated with levels and weights, while soft clauses can be only associated with weights.
However, levels can be simulated by properly modifying weights as follows:
\begin{enumerate}
\item
Let $s := 1 + \sum_{r \in \W^1}{\mathit{weight}(r)}$, and let $l$ be the smallest integer such that $l \geq 2$ and $\W^l \neq \emptyset$.

\item
For all weak constraints $r$ in $\W^l$, set $\mathit{weight}(r) := \mathit{weight}(r) \cdot s$ and $\mathit{level}(r) := 1$.

\item
Repeat steps 1--2 while there are weak constraints of level different from 1.
\end{enumerate}
Actually, the above process was reverted in the MaxSAT literature in order to take advantage from the ``identification of levels'' \cite{DBLP:conf/ijcai/ArgelichLS09}.
In fact, taking into account levels is crucial for some algorithms, as for example linear search sat-unsat (i.e., Algorithm~\ref{alg:linear}).

\begin{example}
Consider the following simple program $\Pi_2$, for any $m \geq 1$ and any $n \geq 1$:
\begin{equation*}
    \begin{array}{lll}
        p_{ij} \leftarrow \naf\naf p_{ij} \qquad&
        1@i \cla \naf p_{ij} \qquad& 
        \forall i \in [1..m], \forall j \in [1..n] \\
    \end{array}
\end{equation*}
In the worst case, Algorithm~\ref{alg:linear} has to find $n$ stable models before understanding that the optimum cost of level $m$ is 0.
Similarly for the other levels, for a total of $m \cdot n$ stable models in the worst case.
If levels are ignored, then in the worst case the upper bound is improved by one after each stable model search, and therefore exponentially many stable models are computed.
\hfill$\blacksquare$
\end{example}

The first algorithms based on unsatisfiable core analysis were introduced in MaxSAT \cite{DBLP:conf/sat/FuM06,DBLP:conf/sat/Marques-SilvaM08,DBLP:conf/date/Marques-SilvaP08,DBLP:conf/sat/ManquinhoSP09,DBLP:conf/sat/AnsoteguiBL09}, and subsequently ported to ASP in the experimental solver \textsc{unclasp} \cite{DBLP:conf/iclp/AndresKMS12}, where also \textsc{oll} was presented.
The main drawback of \textsc{oll} is that it may add new constraints if soft literals introduced by the analysis of previous cores belong to subsequently detected unsatisfiable cores, and these new constraints only minimal differ from those already added to the input program.
This drawback was later circumvented in the MaxSAT solver \textsc{mscg} using a smart encoding based on sorting networks \cite{DBLP:conf/cp/MorgadoDM14}.

Algorithm \textsc{oll} was also implemented in the ASP solver \textsc{wasp} \cite{Alviano01092015}, and compared with other algorithms from MaxSAT literature such as \textsc{pmres} \cite{DBLP:conf/aaai/NarodytskaB14}.
These two algorithms, \textsc{oll} and \textsc{pmres}, are also the origin of two other algorithms for MaxSAT, namely \textsc{one} and \textsc{k} \cite{DBLP:conf/ijcai/AlvianoDR15}.
In fact, \textsc{one} is a simplification of \textsc{oll}, and adds exactly one constraint for each analyzed core;
whether new soft literals will be part of other unsatisfiable cores is irrelevant for \textsc{one}, as all the information required to complete the computation is already encoded in the added constraint.

Concerning \textsc{k}, instead, it is a generalization of \textsc{one} based on the observation that aggregates involving a huge number of literals are handled inefficiently in the refutation process.
This fact was first observed by \citeANP{DBLP:conf/aaai/NarodytskaB14}, who proposed \textsc{pmres}:
the inefficiency is circumvented because \textsc{pmres} introduces constraints of size 3.
The more general idea underlying \textsc{k} is to bound the size of the new constraints to a given constant $k$.
Hence, \textsc{k} combines \textsc{one} and \textsc{pmres}, which are obtained for special values of $k$ (respectively, infinity and 3).
%

Finally, it is important to observe that current core based algorithms in ASP cannot continuously improve both lower and upper bound.
Alternatives for obtaining such a behavior are based on combinations of different algorithms, either sequentially or in parallel \cite{Alviano01092015,DBLP:conf/lpnmr/GebserKK0S15}.
The strategy suggested in this paper is instead to improve core based algorithms by shrinking unsatisfiable cores.
Within this respect, the proposed shrinking procedures are also original, and the closest approaches in MaxSAT iteratively remove one literal from the unsatisfiable core, either obtaining a smaller unsatisfiable core, or a necessary literal in the processed unsatisfiable core \cite{DBLP:conf/fmcad/Nadel10,DBLP:journals/jsat/NadelRS14}.

\section{Conclusion}

The combination of ASP programs and weak constraints is important to ease the modeling of optimization problems.
However, the computation of optimum stable models is often very hard, and suboptimal stable models may be the only affordable solutions in some cases.
Despite that fact, the most efficient algorithms developed so far, which are based on unsatisfiable core analysis, produce few (or even no) stable models while running to completion.

A concrete strategy to improve current ASP solvers is presented in this paper:
better and better stable models can be produced if unsatisfiable cores are shrunk.
The overhead due to the shrinking process is limited by introducing a budget on the running time, and eventually a performance gain is obtained thanks to the reduced size of the analyzed unsatisfiable cores.
On the instances of the Sixth ASP Competition, our implementation is often able to provide (suboptimal) stable models with a guarantee of distance to the optimum cost of around $10\%$.

\section*{Acknowledgement}

This work was partially supported by the National Group for Scientific Computation (GNCS-INDAM), by the Italian Ministry of University and Research under PON project ``Ba2Know (Business Analytics to Know) Service Innovation - LAB'', No.\ PON03PE\_00001\_1, and by the Italian Ministry of Economic Development under project ``PIUCultura (Paradigmi Innovativi per l'Utilizzo della Cultura)'' n.\ F/020016/01--02/X27.

\bibliographystyle{acmtrans}

\label{lastpage}

\newenvironment{reviewer}
{\em\noindent\ignorespaces}{\ignorespacesafterend \medskip }

\newenvironment{answer}
{\color{blue}\noindent\ignorespaces}{\ignorespacesafterend \bigskip }

\end{document}